\documentclass[12pt,reqno]{amsproc}

\usepackage[margin=1in]{geometry}

\usepackage[T1]{fontenc}
\usepackage{bbm}
\usepackage[usenames]{color}
\usepackage{mathpazo}
\usepackage{amsmath,amssymb,amsthm}
\usepackage{wrapfig}
\usepackage{tikz-cd}
\usepackage[shortlabels]{enumitem}

\usepackage{caption}
\usepackage{mathrsfs,bm,nicefrac}
\usepackage[all,cmtip]{xy}
\usepackage{subcaption}

\newcommand{\orcid}[1]{\href{https://orcid.org/#1}{\includegraphics[width=10pt]{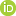}}}

\definecolor{PineGreen}{rgb}{0.0,0.47,0.44}
\definecolor{MidnightBlue}{rgb}{0.1,0.1,0.44}
\definecolor{magenta}{rgb}{1.0,0.0,1.0}
\definecolor{bl1}{HTML}{4479A1}
\definecolor{pur1}{HTML}{52196D}
\definecolor{mag1}{HTML}{2AD0F1}
\definecolor{org1}{rgb}{.92,.39.21}
\definecolor{pur2}{rgb}{.53,.47,.7}
\definecolor{desycyan}{rgb}{0.00,0.68,0.93}
\definecolor{desyorange}{rgb}{0.96,0.52,0.07}
\definecolor{desygray}{rgb}{0.47,0.47,0.47}

\usepackage{hyperref}
\hypersetup{
	colorlinks=true,
	linkcolor=desycyan,
	citecolor=PineGreen,
	filecolor=magenta,
	urlcolor=desyorange
}

\usepackage[labeled,resetlabels]{multibib}
\newcites{P}{Publications}

\makeatletter
\newcommand{\eqnum}{\refstepcounter{equation}\textup{\tagform@{\theequation}}}
\makeatother

\newtheorem{theorem}{Theorem}
\numberwithin{theorem}{section}
\newtheorem{proposition}[theorem]{Proposition}
\newtheorem*{theorem*}{Theorem}
\newtheorem{lemma}[theorem]{Lemma}
\newtheorem{corollary}[theorem]{Corollary}
\theoremstyle{definition}

\newtheorem{definition}[theorem]{Definition}

\theoremstyle{remark}
\newtheorem{remark}[theorem]{Remark}
\newtheorem{example}[theorem]{Example}
\AtEndEnvironment{example}{\null\hfill\ensuremath{\diamond}}

\AtEndEnvironment{remark}{\null\hfill\ensuremath{\spadesuit}}

\newcommand{\Con}{\mathbf{Con}}

\newcommand{\RR}{\mathbb{R}}

\newcommand{\QQ}{\mathbb{Q}}

\newcommand{\CC}{\mathbb{C}}
\newcommand{\ZZ}{\mathbb{Z}}
\newcommand{\NN}{\mathbb{N}}

\newcommand{\scrH}{\mathscr{H}}

\newcommand{\cA}{\mathcal{A}}

\newcommand{\cU}{\mathcal{U}}

\newcommand{\cF}{\mathcal{F}}
\newcommand{\cG}{\mathcal{G}}
\newcommand{\cH}{\mathcal{H}}
\newcommand{\cI}{\mathcal{I}}

\newcommand{\frakr}{\mathfrak{r}}
\newcommand{\frakR}{\mathfrak{R}}

\newcommand{\newt}{\mathbf{N}}

\newcommand{\rank}{\mathrm{rank}}
\newcommand{\vol}{\mathrm{vol}}
\newcommand{\conv}{\mathrm{conv}}

\newcommand{\xx}{\boldsymbol{x}}

\makeatletter
\DeclareRobustCommand
\myvdots{\vbox{\baselineskip4\p@ \lineskiplimit\z@
		\hbox{.}\hbox{.}\hbox{.}}}
\makeatother

\usepackage[foot]{amsaddr}

\newcommand{\avg}[1]{\left< #1 \right>} 

\usepackage{tikz}
\usepackage{tikz-feynman}

\def\DD{D\kern-.7em\raise0.3ex\hbox{\char '55}\kern.33em}

\usepackage{booktabs}
\usepackage{colortbl}
\definecolor{Ftitle}{RGB}{11,46,108}
\colorlet{tableheadcolor}{Ftitle!25} 
\colorlet{tablerowcolor}{gray!10} 
\title{Characterizing Cohen-Macaulay One-Loop Feynman Integrals
\\
 }

\author{Kyrill Michaelsen\orcid{0009-0009-1681-2162}$\,{}^1$}
\author{Felix Tellander\orcid{0000-0001-6418-8047}$\,{}^{2,3}$}
\email{\href{mailto:kyrill.michaelsen@uni-hamburg.de}{kyrill.michaelsen@uni-hamburg.de}}
\email{\href{mailto:felix@tellander.se}{felix@tellander.se}}
\address[1]{University of Hamburg, Department of Mathematics, 20146 Hamburg, Germany}
\address[2]{Mathematical Institute, University of Oxford, Oxford OX2 6GG, UK}
\address[3]{School of Mathematics and Hamilton Mathematics Institute, Trinity College,\newline
\hspace*{1.15cm}Dublin 2, Ireland}
\date{\today}
\begin{document}
\begin{abstract}
We study the generalized hypergeometric systems, in the sense of Gel'fand, Kapranov, and Zelevinsky, associated with one-loop Feynman integrals, and determine when their rank is independent of space-time dimension and propagator powers. This is equivalent to classifying when the associated affine semigroup ring is Cohen-Macaulay. For massive one-loop integrals, we prove necessary and sufficient conditions for Cohen-Macaulayness, generalizing previous results on normality for these rings. We show that for Feynman integrals, the Cohen-Macaulay property is fully determined by an integer linear program built from the Newton polytope of the integrand and find a graphical description of its solutions. Furthermore, we provide a sufficient condition for Cohen-Macaulayness of general one-loop integrals.
\end{abstract}
\maketitle
\tableofcontents
\section{Introduction}
Feynman integrals lie at the core of computing particle interactions in high-energy physics \cite{Weinberg:1995mt}, the scattering of gravitational waves \cite{Bern:2019nnu,Dlapa:2022lmu}, or determining critical exponents in statistical field theory \cite{Zinn-Justin:1998qip}. While they have been the subject of immense research efforts for decades \cite{Weinzierl:2022eaz}, evaluating them is still highly non-trivial. It is therefore crucial to better understand the mathematical structures underpinning Feynman integrals, to facilitate both better numerical \cite{Smirnov:2021rhf,Heinrich:2023til,Borinsky:2020rqs,Borinsky:2023jdv} and analytic \cite{Smirnov:2012gma,Panzer:2014caa,Chestnov:2022alh,Ananthanarayan:2022ntm} evaluation of them. Since Feynman integrals of fields with tensor structure can always be reduced to scalar Feynman integrals \cite{PASSARINO1979151, Tarasov:1996br,Tarasov:1997kx}, the understanding of scalar Feynman integrals is crucial for the calculation of any Standard Model process. In Lee-Pomeransky \cite{Lee:2013hzt} and Feynman parametrization \cite{Itzykson:1980rh}, scalar integrals are of Euler-Mellin type \cite{delaCruz:2019skx,Klausen:2019hrg} in the sense of \cite{Nilsson2013, Berkesch2014}, which are solutions to \textit{Gel’fand-Kapranov-Zelevinsky} (GKZ) systems \cite{Gelfand1989,Gelfand1990}. These are systems of linear partial differential equations generalizing the classical hypergeometric equations. 

In general, a Feynman integral has fewer kinematic variables than required by the GKZ framework. However, for one-loop integrals, it is often the case that the number of variables match perfectly and the GKZ framework applies without any modifications. Even more so, the GKZ theory has been used successfully for higher loop integrals, especially in the banana family \cite{Klemm:2019dbm,Bonisch:2020qmm}. Within the vast theory of generalized hypergeometry, one might wonder: Do Feynman integrals exhibit any special structure? The answer to this is affirmative; many Feynman integrals have normal affine semigroups \cite{Tellander:2021xdz,walther2022feynman,Dlapa:2023cvx}. This has a great simplifying effect on the GKZ system, as it implies that the associated affine semigroup ring is \emph{Cohen-Macaulay} \cite{Hochster1972}. This means, in particular, that the dimension of the solution space (i.e. the number of master integrals) and the characteristic variety (i.e. the microlocal description of the Landau singularities) are both independent of space-time dimension and propagator powers \cite{SST,Matusevich2005}. For computational purposes, this implies that the number of master integrals is simply given by the simplicial volume of a polytope and that the characteristic variety (even characteristic cycle) is fully described by the \emph{principal A-determinant} \cite{adolphson1994hypergeometric} (see also \cite{Klausen:2021yrt} for an introduction for physicists) and therefore fully calculable using classical lattice and commutative Gr\"obner basis methods.

In \cite{Matusevich2005} the authors showed that the rank (i.e. dimension of the the solution space) of a GKZ system is independent of the parameters (i.e. space-time dimension and propagator powers for Feynman integrals) if and only if the associated semigroup ring is Cohen-Macaulay. In order to harness the full power of this result, we want to describe when the semigroup ring is Cohen-Macaulay and not necessarily normal. In this paper, we give a necessary and sufficient condition for Cohen-Macaulayness for massive one-loop integrals, Theorem~\ref{massivetheorem}, and a sufficient condition for general one-loop integrals, Theorem~\ref{generaltheorem}. 

This paper is organized as follows. In Section \ref{sec: generalized hypergeometric systems} we describe how GKZ hypergeometric systems arise from Feynman integrals in Lee-Pomeransky form and in Section \ref{sec: affine semigroups} we introduce the notions of affine semigroup, semigroup rings, and Cohen-Macaulay rings. Important for our study is that one-loop Feynman integrals fit within the framework of edge rings of \cite{ohsugi1998normal}, we introduce this framework in Section \ref{sec: graphs and edge polytopes} and use this in Proposition \ref{prop: (ii) always satisfied} to show how Cohen-Macaulayness for one-loop integrals simplifies compared to the general case. In Section \ref{sec: massive one loop} we prove three results on the Cohen-Macaulayness of massive one-loop integrals; the main result, Theorem~\ref{massivetheorem}, is a full classification. En route, we give a graphical description of the integer linear programs, that determine Cohen-Macaulayness. Finally, in Section \ref{sec: general one-loop}, we address this question for general one-loop integrals and prove a sufficient condition for them to be Cohen-Macaulay in Theorem~\ref{generaltheorem}.
\section{Generalized hypergeometric systems}\label{sec: generalized hypergeometric systems}
In this section, we introduce the notion of GKZ-systems to demonstrate how they allow for understanding Feynman integrals as $A$-hypergeometric functions. Feynman integrals are integrals associated to a \emph{Feynman graph} $G$. Let $G:=(E,V)$ be a connected bridge-less graph with edge set $E$, vertex set $V$ and loop number $L=|E|-|V|+1$. The vertex set $V$ has the disjoint partition $V=V_{\mathrm{ext}}\sqcup V_{\mathrm{int}}$ where we assume $V_{\mathrm{ext}}$ is non-empty and each vertex $v\in V_{\mathrm{ext}}$ is assigned an external incoming momentum $p_v$. Each edge $e\in E$ is assigned a mass $m_e$. Writing the Feynman integral in Lee-Pomeransky form \cite{Lee:2013hzt}
using analytic regularization \cite{Speer:1968qxh} the integral has the form of an Euler-Mellin integral \cite{Nilsson2013,Berkesch2014}:
\begin{equation}\label{eq: LP integrals}
    \cI(z;\,D,\nu_1,\ldots,\nu_{|E|}):=\int_{\RR_{\ge 0}^{|E|}}\frac{x_1^{\nu_1}\cdots x_{|E|}^{\nu_{|E|}}}{\mathcal{G}(z;\,x)^{D/2}}\frac{dx_1}{x_1}\wedge\cdots\wedge\frac{dx_{|E|}}{x_{|E|}}
\end{equation}
where $D,\nu_1,\ldots,\nu_{|E|}\in\CC$ and $\cG(z;\,x):=\cU(x)+\cF(z;\,x)$ is the sum of the two \emph{Symanzik polynomials} where $z$ represents the kinematic variables. These polynomials have a simple combinatorial description in terms of the underlying graph \cite{Nakanishi1971}:
\begin{equation}\label{eq:F_Symanzik}
\begin{split}
    \mathcal{U}&:=\sum_{\substack{~~~\,T {\rm \;a \; spanning} \\ {\rm tree \; of \; }G}}\;\;\;\prod_{e\not\in T}x_e,\\
    \mathcal{F}&:= \mathcal{U}\sum_{e\in E}m_e^2x_e- \sum_{\substack{F {\rm \;a \; spanning} \\ {\rm 2-forest \; of \; }G}}\;\;\; p(F)^2\prod_{e\not\in F}x_e,
\end{split}
\end{equation}
where a \emph{spanning tree} is a connected subgraph of $G$ containing all vertices but no cycles, and a \emph{spanning two-forest} is a spanning tree with one more edge removed. For a spanning two-forest $F=(T,T')$ of $G$ we define $p(F)=\sum_{v\in T\cap V_{\mathrm{ext}}}p_v$ as the total momentum flowing through the cut of $G$. With multi-index notation $x^{\boldsymbol{\alpha}{_i}}=x_1^{\alpha_{1i}}x_2^{\alpha_{2i}}... x_{|E|}^{\alpha_{|E|i}}$ the Lee-Pomeransky polynomial can be written as:
\begin{equation}
    \mathcal{G}(z;\,x)=\sum_{i=1}^nz_ix^{\boldsymbol{\alpha}_i}
\end{equation} 
which allows for storing the information of the Euler-Mellin integral inside a $(|E|+1)\times n$ integer matrix $A$ and a vector $\beta\in\CC^{|E|+1}$ \cite{Nilsson2013,Berkesch2014}: 
    \begin{equation}\label{Amatrix}
        A :=\left(\mathbf{a}_1\cdots \mathbf{a}_n\right) =\begin{pmatrix}
        1 & 1 & \cdots & 1 \\
        \alpha_{11} & \alpha_{12} & \cdots & \alpha_{1n} \\
        \alpha_{21} & \alpha_{22} & \cdots & \alpha_{2n} \\
        \vdots&\vdots&\ddots&\vdots\\
        \alpha_{|E|1} & \alpha_{|E|2} & \cdots & \alpha_{|E|n}
        \end{pmatrix},\quad\beta =
        \begin{pmatrix}
        -D/2 \\ 
        -\nu_1 \\ 
        \vdots \\
        -\nu_{|E|}
        \end{pmatrix}.
    \end{equation}

We define the GKZ-system as a left ideal of the Weyl algebra $W:=\QQ(\beta)[z_1\dots,z_n]\avg{\partial_{1},\dots,\partial_{n}}$ consisting of the toric ideal $I_A:=\avg{\partial^\mathbf{v}-\partial^\mathbf{u}\,|\,A\mathbf{v}=A\mathbf{u}, \,\forall\mathbf{v},\mathbf{u}\in\NN^n}$ and the ideal $Z_A(\beta)$ of Euler operators:
\begin{equation}
   Z_A(\beta):= \avg{\sum_{j=1}^{n}a_{ij}z_j\partial_{j}-\beta_i\,|\,i\in\{1,\dots,|E|+1\}}.
\end{equation}
The GKZ-hypergeometric system $H_A(\beta)\subset W$ is the left ideal defined by the sum: 
    \begin{equation}
        H_A(\beta):=I_A +Z_A(\beta).
    \end{equation}
 
The ideal $H_A(\beta)$ is always holonomic \cite{adolphson1994hypergeometric}, meaning in particular that it has finite rank~\cite{SST}. Moreover, for our special $A$-matrix in \eqref{Amatrix} it is obvious that the $(1\cdots 1)$ lie in the $\QQ$-row-span of $A$, meaning that $H_A(\beta)$ is even regular holonomic \cite{hotta1998equivariant}.

 The \emph{Newton polytope} of a polynomial is the convex hull of exponent vectors in the polynomial. For the Lee-Pomeransky polynomial, the Newton polytope $\newt[\mathcal{G}]$ is given by the convex hull of the column vectors of the  $A$-matrix without its first row of ones:
\begin{equation}
  \newt[\cG]=\conv\{\boldsymbol{\alpha}_1,\ldots,\boldsymbol{\alpha}_n\}.
\end{equation}
In the papers \cite{Nilsson2013, Berkesch2014} the authors describe, in terms of the Newton polytope $\newt[\cG]$, when an Euler-Mellin integral like \eqref{eq: LP integrals} is convergent. Moreover, they show that it is a solution to the GKZ system $H_A(\beta)$ with $A$ and $\beta$ given in \eqref{Amatrix}. For generic $\beta$ it was proven in \cite{Gelfand1989,SST} that the rank of $H_A(\beta)$ is given by the simplicial volume of $\newt[\cG]$. It was later shown in \cite{Matusevich2005} that this is a characteristic of Cohen-Macaulayness:
\begin{theorem}[{\cite[Corollary 9.2]{Matusevich2005}}]\label{thm: rank=vol}
    \begin{equation*}
        \rank(H_A(\beta))=|E|!\,\vol(\newt[\cG])\,\forall\beta\iff \CC[\partial_1,\ldots,\partial_n]/I_A\ \mathrm{is}\ \mathrm{Cohen}\text{-}\mathrm{Macaulay}.
    \end{equation*}
\end{theorem}
The above theorem is one of the main reasons to study the Cohen-Macaulay property of GKZ systems and Feynman integrals. However, Cohen-Macaulayness also affects the singularities, as we will now describe.

The \emph{singular locus} (or Landau singularities) of a GKZ system is given by the \emph{principal A-determinant} \cite[Chapter 10]{gelfand2008discriminants}:
\begin{equation}\label{eq:EaFact}
    E_A(\cG):=\pm\prod_{\Gamma\subseteq \newt[\cG]}\Delta_{A\cap\Gamma}(\cG)^{\mu_\Gamma}
\end{equation}
where $\Delta_{A\cap\Gamma}(\cG):=\Delta_{A\cap\Gamma}(\left.\cG\right|_\Gamma)$ is the \emph{A-discriminant} of the polynomial $\left.\cG\right|_\Gamma$ which is the coordinate restriction of $\cG$ supported on $\Gamma$. The powers $\mu_\Gamma$ are positive integers given by the multiplicity of the toric variety defined by $A$ along the toric variety defined by $A\cap\Gamma$ and are calculable as lattice indices. Under the Cohen-Macaulay assumption, the principal $A$-determinant describes the full characteristic cycle \cite{Gelfand1989,adolphson1994hypergeometric}:
\begin{equation}
    CC(H_A(\beta))=\rank(H_A(\beta))\Con(0)+\sum_{\Gamma\subseteq\newt[\cG]}\mu_\Gamma\Con(\Delta_{A\cap\Gamma}(\cG))
\end{equation}
where $\Con(\cdot)$ denotes the conormal variety \cite{gelfand2008discriminants}. The characteristic cycle provides the full microlocal description of the singularities. When the system is not a GKZ system, or when the semigroup ring is not Cohen-Macaulay, calculating the characteristic cycle is a challenging problem. It is only very recently that general algorithms for doing this have been constructed \cite{Helmer:2025yuf,Matsubara-Heo:2025lrq,Helmer:2025ljj}.

We conclude this section with an example of the massive triangle Feynman integral:
\begin{example}
For the triangle graph in Figure \ref{fig: triangle graph}, the two Symanzik polynomials in \eqref{eq:F_Symanzik} are:
\begin{equation}
\begin{split}
\mathcal{U}(x)&=x_1+x_2+x_3,\\
\mathcal{F}(p_i^2,m_i^2;\,x)&
= (x_1+x_2+x_3)( m_1^2x_1 + m_2^2x_2 + m_3^2x_3)-p_1^2x_1x_2
- p_2^2x_2x_3
- p_3^2x_1x_3,
\end{split}
\end{equation}
which yield the following form of the $A$-matrix:
\begin{equation*}
    A=\begin{pmatrix}
        1&1&1&1&1&1&1&1&1\\
        1&0&0&1&1&0&2&0&0\\
        0&1&0&1&0&1&0&2&0\\
        0&0&1&0&1&1&0&0&2
    \end{pmatrix}.
\end{equation*}
It follows from e.g. \cite[Theorem 1.1]{Tellander:2021xdz} or \cite[Theorem 5.4]{Dlapa:2023cvx} that the semigroup generated by $A$ is normal, and hence Cohen-Macaulay. By Theorem \ref{thm: rank=vol} the rank of this GKZ system equals the simplicial volume of the Newton polytope, depicted in Figure \ref{fig: newton polytope}, for every possible $\beta$. In particular we get $\mathrm{rank}(H_A(\beta))=3!\cdot 7/6=7$. This coincides with the master integrals in physics: the triangle itself, three different bubbles (obtained from contracting one edge), and three different tadpoles (obtained from contracting two edges).

The discriminants contributing to the principal $A$-determinant in \eqref{eq:EaFact} are listed in e.g. \cite[Section 4.2]{Dlapa:2023cvx}. The multiplicities $\mu_\Gamma$ are equal to one for every face in this example.
\begin{figure}[t] 
    \begin{subfigure}[t]{0.45\textwidth}
    \centering 
    \vspace*{-3.86cm}
       \begin{tikzpicture}
\begin{feynman}
\vertex (e){\(p_1\)};
\vertex [right = of e] (a);
\vertex [below right = of a] (b);
\vertex [above right = of a] (c);
\vertex [below right = of b] (e1){\(p_2\)};
\vertex [above right = of c] (e2){\(p_3\)};
\diagram* {
	(e) -- [fermion] (a) -- [edge label'=\(x_2\)] (b) -- [edge label'=\(x_3\)] (c) --[edge label'=\(x_1\)] (a),
	(e1) -- [fermion] (b),
	(e2) -- [fermion] (c),
};
\end{feynman}
\end{tikzpicture}
\caption{}\label{fig: triangle graph}
    \end{subfigure}\hspace*{0.5cm}
    \begin{subfigure}[t]{0.45\textwidth}
    \centering
       \begin{tikzpicture}%
	[x={(-1cm, -0.395971cm)}, 
	y={(1cm, -0.401902cm)}, 
	z={(0.017418cm, 1cm)},
	scale=2.000000,
	back/.style={loosely dotted, thin},
	edge/.style={color=black, thick},
	facet/.style={fill=green,fill opacity=0.200000},
	vertex/.style={inner sep=1pt,circle,draw=blue!25!black,fill=blue!75!black,thick},
    baseline=-\the\dimexpr\fontdimen22\textfont2\relax]

\coordinate (0.00000, 0.00000, 1.00000) at (0.00000, 0.00000, 1.00000);
\coordinate (0.00000, 0.00000, 2.00000) at (0.00000, 0.00000, 2.00000);
\coordinate (0.00000, 1.00000, 0.00000) at (0.00000, 1.00000, 0.00000);
\coordinate (2.00000, 0.00000, 0.00000) at (2.00000, 0.00000, 0.00000);
\coordinate (0.00000, 2.00000, 0.00000) at (0.00000, 2.00000, 0.00000);
\coordinate (1.00000, 0.00000, 0.00000) at (1.00000, 0.00000, 0.00000);
\coordinate (0.00000, 0.00000, 0.00000) at (0.00000, 0.00000, 0.00000);
\draw[edge,back] (0.00000, 0.00000, 1.00000) -- (0.00000, 0.00000, 2.00000);
\draw[edge,back] (0.00000, 0.00000, 1.00000) -- (0.00000, 1.00000, 0.00000);
\draw[edge,back] (0.00000, 0.00000, 1.00000) -- (1.00000, 0.00000, 0.00000);
\draw[edge,back] (0.00000, 1.00000, 0.00000) -- (0.00000, 2.00000, 0.00000);
\draw[edge,back] (0.00000, 1.00000, 0.00000) -- (1.00000, 0.00000, 0.00000);
\draw[edge,back] (2.00000, 0.00000, 0.00000) -- (1.00000, 0.00000, 0.00000);
\node[vertex] at (0.00000, 1.00000, 0.00000)     {};
\node[vertex] at (1.00000, 0.00000, 0.00000)     {};
\node[vertex] at (0.00000, 0.00000, 1.00000)     {};
\node[vertex] at (0.00000, 1.00000, 1.00000)     {};
\node[vertex] at (1.00000, 1.00000, 0.00000)     {};
\node[vertex] at (1.00000, 0.00000, 1.00000)     {};

\fill[facet] (0.00000, 2.00000, 0.00000) -- (0.00000, 0.00000, 2.00000) -- (2.00000, 0.00000, 0.00000) -- cycle {};
\draw[edge] (0.00000, 0.00000, 2.00000) -- (2.00000, 0.00000, 0.00000);
\draw[edge] (0.00000, 0.00000, 2.00000) -- (0.00000, 2.00000, 0.00000);
\draw[edge] (2.00000, 0.00000, 0.00000) -- (0.00000, 2.00000, 0.00000);

\node[vertex] at (0.00000, 0.00000, 2.00000)     {};
\node[vertex] at (2.00000, 0.00000, 0.00000)     {};
\node[vertex] at (0.00000, 2.00000, 0.00000)     {};
\end{tikzpicture}
\caption{}\label{fig: newton polytope}
\end{subfigure} 
    \caption{Triangle Feynman graph (A) and its Newton polytope of of volume $7/6$ (B).}
\end{figure}
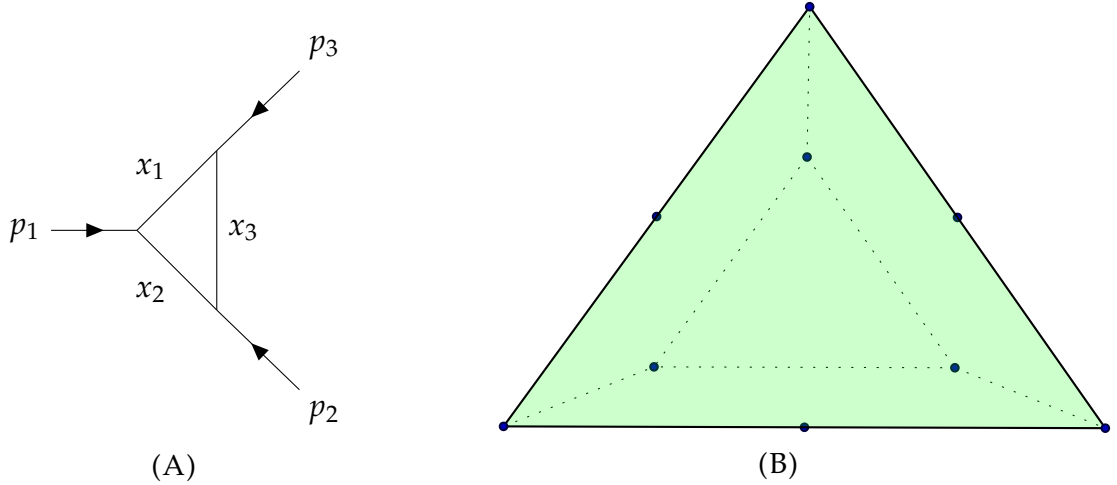
\end{example}
For more background on applying the GKZ framework to Feynman integrals in Lee-Pomeransky parametrisation, we refer to \cite{delaCruz:2019skx, Klausen:2019hrg, Tellander:2021xdz, Ananthanarayan:2022ntm,Munch:2022ouq,Dlapa:2023cvx}.

\section{Affine semigroups}\label{sec: affine semigroups}
We define an \emph{affine semigroup} as a finitely generated submonoid of $\ZZ^d$. For a subset $A=\{\mathbf{a}_1,\ldots,\mathbf{a}_n\}\subset\ZZ^d\setminus\{0\}$ we denote by $Q:=\NN A$ the affine semigroup containing all non-negative linear integer combinations of the elements in $A$. With abuse of notation, we will also denote by $A=(\mathbf{a}_1 \cdots \mathbf{a}_n)$ the $(d\times n)$-matrix with columns $\mathbf{a}_i$. Further, we will assume that $A$ is pointed, i.e. the origin is the only linear subspace of the underlying cone $\RR_{\ge 0}A=\{\sum_{i=1}^nc_i\mathbf{a}_i\,|\,c_i\in\RR_{\ge 0}\}$. A subset $F$ of $A$ is called a \emph{face} if $\RR_{\ge 0}F$ is a face of the cone $\RR_{\ge 0}A$. Similarly, we will refer to $F$ as a face of the semigroup $\NN A$ to indicate that $\NN F$ is a face of $\NN A$.
\begin{definition}\label{def: normal}
    A semigroup $\NN A$ is said to be \emph{normal} if it is saturated, i.e. if
    \begin{equation}
        \NN A=\ZZ A\cap \RR_{\ge 0}A.
    \end{equation}
\end{definition}
\noindent If an affine semigroup $Q$ is not normal, it has a non-empty set of \emph{holes}: $\cH(Q):=(\ZZ A\cap \RR_{\ge 0} Q)\setminus Q$.

Let $k$ be an algebraically closed field, the \emph{affine semigroup ring} 
\begin{equation}
 k[Q]=k[t^{\mathbf{a}_1},\ldots,t^{\mathbf{a}_n}]:=\bigoplus_{\mathbf{a}\in\NN A}k\cdot t^{\mathbf{a}}   
\end{equation}
is a subring of the the ring of Laurent polynomials $k[t_1^{\pm 1},\ldots,t_d^{\pm 1}]$. Let $R=k[y_1,\ldots,y_n]$ be a polynomial ring and $I_A=\avg{y^\mathbf{u}-y^{\mathbf{v}}\,|\, A\mathbf{u}=A\mathbf{v},\ \mathbf{u},\mathbf{v}\in\NN^n}\subseteq R$ the toric ideal associated to $A$, and $S_A=R/I_A$ its quotient ring. The quotient ring $S_A$ is isomorphic to the affine semigroup ring of $A,\ S_A\simeq k[\NN A],$ where multiplication in $k[\NN A]$ is given by semigroup addition of the exponents.

A \emph{free resolution of} $S_A$ \emph{of length} $l$ is a degree-preserving exact sequence of finitely generated $R$-modules:
\begin{equation}
    0\longleftarrow R/I_A\longleftarrow R=F^0\longleftarrow F^1\longleftarrow\cdots\longleftarrow F^l\longleftarrow 0
\end{equation}
such that every $F^j$ is a free $R$-module. A free resolution is \emph{minimal} when the ranks of $F^j$ are taken to be as small as possible. It follows directly from Hilbert's syzygy theorem and Auslander-Buchsbaum's formula (see e.g.~\cite{Eisenbud1995}) that the length of a minimal free resolution is bounded from above and below:
\begin{equation}
    n-d\le l\le n.
\end{equation}
When the lower bound is saturated, the ideal (and the module $R/I$) is said to be \emph{Cohen-Macaulay}, for more background, see \cite{Bruns1993}.
\begin{definition}[Cohen-Macaulay]
A homogeneous ideal $I$ (and the module $R/I$) is said to be \emph{Cohen-Macaulay} if the length $l$ of the minimal free resolution of $R/I$ is $n-d$.
\end{definition}
\begin{remark}
A toric ideal $I_A$ is homogeneous if and only if $(1,\ldots,1)$ is in the $\QQ$-row-span of $A$.
\end{remark}
We now recall a famous theorem of Hochster \cite{Hochster1972}:
\begin{theorem}\label{thm: Hochster}
    If the semigroup $Q$ is normal, then $k[Q]$ is Cohen-Macaulay.
\end{theorem}

Let $F_1,\ldots,F_m$ be the facets of $Q$ and we denote by $\widetilde{Q}:=\cap_{i=1}^m(Q-\NN F_i)$. For every non-empty subset $J\subset\{1,\ldots,m\}$ set
\begin{equation}
    G_J=\bigcap_{i\notin J}(Q-\NN F_i)\setminus\bigcup_{j\in J}(Q-\NN F_j).
\end{equation}
We let $\pi_J$ be the simplicial complex of non-empty subset $I$ of $J$ such that $\cap_{i\in I}F_i$ is a non-empty face of $Q$. A complex $\pi_J$ is \emph{acyclic} if its reduced homology group is zero for all indices. We now state the Main Theorem in \cite{Trung1986}:
\begin{theorem}\label{thm: Cohen-Macaulay Trung Hoa}
    Let $Q$ be an arbitrary affine semigroup. Then $k[Q]$ is Cohen-Macaulay if and only if the following conditions are satisfied:
    \begin{itemize}
        \item[(i)] $\widetilde{Q}=Q$, and
        \item[(ii)] for every non-empty proper subset $J$ of $\{1,\ldots,m\}$, $G_J=\emptyset$ or $\pi_J$ is acyclic.
    \end{itemize}
\end{theorem}
Much more recently, a slightly modified version of this theorem was proven in \cite{Matusevich2024}, which we recall below, and this will be the version of this theorem we actually use. 

Let $\cA=\{\scrH_i\}_{i=1}^m$ be a finite collection of hyperplanes passing through the origin. A \emph{region} in $\RR^d$ is a connected component $\frakr$ of $\RR^d\setminus\bigcup_{\scrH\in\cA}\scrH$ and the set containing all regions over $\cA$ is denoted by $\frakr(\cA)$. For any subset $J$ of $\{1,\ldots,m\}$ we define the \emph{cumulative region} as
\begin{equation}
    \frakR_J:=\bigcup_{i\in J}\scrH^+_i
\end{equation}
where $\scrH^+$ denotes the positive half-space including its boundary. We are now ready to rephrase Theorem \ref{thm: Cohen-Macaulay Trung Hoa}:
\begin{theorem}[{\cite[Theorem 6.3]{Matusevich2024}}] \label{thm: Cohen-Macaulay Matusevich Yu}
Let $\cA$ be the collection of facet-defining hyperplanes $\scrH_1,\ldots,\scrH_m$ of $Q$. Then $k[Q]$ is Cohen-Macaulay if and only if;
    \begin{itemize}
        \item[(i)] $\widetilde{Q}=Q$, and
        \item[(ii)'] for every $J\subseteq\{1,\ldots,m\}$ with $\frakR_J\in\frakr(\cA)$, $\pi_J$ is acyclic.
    \end{itemize}
\end{theorem}

\section{Graphs and edge polytopes}\label{sec: graphs and edge polytopes}
The semigroup rings arising from one-loop Feynman integrals carry a natural description as \emph{edge rings} in the sense of \cite{ohsugi1998normal}, this description allowed for the complete classification of normality in \cite{Dlapa:2023cvx}. For completeness we describe some of the necessary constructions from this paper we use later.

\begin{definition}\label{edge matrix}[Edge Matrix] Let $H=(E,V)$ be a finite connected graph with vertex set $V=\{0,\ldots,d\}$. If $e=\{i,j\}$ is an edge of $H$ joining vertices $i$ and $j$ we define $\rho(e)\in\RR^{d+1}$ by $\rho(e)=\mathbf{e}_i+\mathbf{e}_j$ where $\mathbf{e}_i$ is the $i$th unit vector in $\RR^{d+1}$. Let $M$ be the matrix whose columns correspond to the finite set $\{\rho(e):\ e\in E\}$, then $M$ is called the \emph{edge matrix} of $H$ and the convex hull of $M$ is called the \emph{edge polytope}.
\label{def:Edge_Matrix}
\end{definition}
To the graph $H$ we associate the subalgebra of $k[t_0,\ldots,t_d]$ generated by all quadratic monomials $t_it_j$ such that $\{i,j\}$ is an edge of $H$ and $t_i^2$ such that $H$ has a loop at vertex $i$. We denote this affine semi group ring by $k[H]$ and call it the \emph{edge ring} of $H$. For this type of rings, \cite{ohsugi1998normal} provides a necessary and sufficient condition for normality in terms of the \emph{odd cycle condition}:
\begin{definition}[Odd cycle condition]\label{def:oddCycle}
A cycle in a graph is called \emph{minimal} if it has no chord and it is said to be \emph{odd} if it is a cycle with odd length. A graph $H$ satisfies the \emph{odd cycle condition} if for two arbitrary minimal odd cycles $C$ and $C'$, either
\begin{itemize}
    \item $C$ and $C'$ have a common vertex, or
    \item there is an edge connecting a vertex of $C$ with a vertex of $C'$.
\end{itemize}
\end{definition}
We may now state one of the main results of \cite{ohsugi1998normal}: 
\begin{theorem}[Corollary 2.3, \cite{ohsugi1998normal}]\label{thm: odd-cycle}
 Let $M$ be the edge matrix of a graph $H$. The semi-group $\NN M$ is normal if and only if the graph $H$ satisfies the odd cycle condition.   
\end{theorem}

Let $\overline{H}$ denote the graph obtained by adding all edges $\{i,j\}$ to $H$ such that $\{i,j\}\notin H$ and $H$ has a loop both at $i$ and at $j$. Let $\overline{M}$ and $M$ denote the respective edge matrices, it follows that $\conv(\overline{M})=\conv(M)$, i.e. the edge polytopes are equal. 
\begin{remark}
If $k[\overline{H}]$ is normal, then $k[\overline{H}]$ is a normalization of $k[H]$ which does not change the underlying polytope.    
\end{remark}

For a fully general one-loop integral with $d$ external legs, the $A$-matrix of the associated GKZ-system will have the form
\begin{equation}
    A=\left(
    \begin{array}{c|c|c}
        1\ \cdots\ 1 & 1\ \ \ \cdots\ \ 1 & 1\ \ \ \cdots\ \ 1 \\
        \mathbb{1}_{d\times d} & \begin{array}{ccc}
            1 && 0\\
            1 && 0\\
            0 && 0\\
            \vdots & \cdots & \vdots\\
            0&&1\\
            0&&1
        \end{array}&  \begin{array}{ccc}
            2 && 0\\
            0 && 0\\
            0 && 0\\
            \vdots & \cdots & \vdots\\
            0&&0\\
            0&&2
        \end{array}
    \end{array}
    \right)
\end{equation}
where it is only the first block with the $d\times d$-identity matrix $\mathbb{1}_{d\times d}$ that appears for every integral. The middle block can be any subset of the $d$-choose-2 possibilities of two ones and the right block can be any of the $d$ possibilities to choose a single two. Of course, to have a non-empty toric geometry and hypergeometric system we need at least $d+2$ columns in total (this corresponds to having at least one scaleless parameter in our Feynman integral). The matrix 
\begin{equation}
    S=\left(\begin{array}{cccc}
         2 & -1 & \cdots &-1  \\
         0&&&\\
         \vdots && \mathbb{1}_{d\times d}&\\
         0&&&
    \end{array}\right)
\end{equation}
is clearly invertible and transform the $A$-matrix to
\begin{equation}\label{eq: edge matrix B}
    B:=SA=\left(
    \begin{array}{c|c|c}
        1\ \cdots\ 1 & 0\ \, \ \ \cdots\ \ 0 & 0\ \, \ \ \cdots\ \ 0 \\
        \mathbb{1}_{d\times d} & \begin{array}{ccc}
            1 && 0\\
            1 && 0\\
            0 && 0\\
            \vdots & \cdots & \vdots\\
            0&&1\\
            0&&1
        \end{array}&  \begin{array}{ccc}
            2 && 0\\
            0 && 0\\
            0 && 0\\
            \vdots & \cdots & \vdots\\
            0&&0\\
            0&&2
        \end{array}
    \end{array}
    \right).
\end{equation}
This matrix as column sum two and only contains non-negative integers, so it is an edge matrix in Definition \ref{def:Edge_Matrix}. Since $S\in\mathrm{GL}(d+1,\QQ)$ the toric ideals satisfy $I_A=I_{B}$ and for the associated hypergeometric system we have $H_A(\beta)=H_{B}(S\beta)$. To each one-loop integral we can thus associate an edge matrix $B=SA$ and use the framework of \cite{ohsugi1998normal}.

\begin{lemma}
    For an edge matrix $B$ of the form \eqref{eq: edge matrix B} the semigroup generated by $\overline{B}$ is normal.
\end{lemma}
\begin{proof}
    Follows directly from \cite[Theorem 5.4]{Dlapa:2023cvx}.
\end{proof}
\begin{corollary}\label{cor: conv equal normal conv}
    The polytopes $\conv(B)$ and $\conv(\overline{B})$ are equal.
\end{corollary}
\begin{proof}
    From \cite[Theorem 5.4]{Dlapa:2023cvx} it follows that the only possible obstructions to normality of these one-loop semigroups come from vertices with loops not connected with an edge. Let $i$ and $j$ be two vertices with a loop and not connected by an edge of the graph with edge matrix $B$, i.e. $2\mathbf{e}_i,2\mathbf{e}_j\in B$ but $\mathbf{e}_i+\mathbf{e}_j\notin B$. Then $\overline{B}$ will contain a vector with a one at position $i$ and $j$ and zeros everywhere else, i.e. $\mathbf{e}_i+\mathbf{e}_j\in\overline{B}$. Since we have the linear combination
    \begin{equation*}
        \frac{1}{2}(2\mathbf{e}_i)+\frac{1}{2}(2\mathbf{e}_j)=\mathbf{e}_i+\mathbf{e}_j
    \end{equation*}
    with both multipliers on the left-hand side equal $1/2$ it follows that $\mathbf{e}_i+\mathbf{e}_j\in\mathrm{conv}(B)$.
\end{proof}
\begin{remark}
    Note that this is equivalent to saying that the facet-defining hyperplanes of $\NN B$ and $\NN\overline{B}$ are the equal. The result of this corollary is also mentioned in the text of~\cite{ohsugi1998normal}.
\end{remark}

\begin{proposition}\label{prop: (ii) always satisfied}
    Let $\NN B$ be the affine semigroup generated by an edge matrix of the form \eqref{eq: edge matrix B}. Then the criterion (ii)' in Theorem \ref{thm: Cohen-Macaulay Matusevich Yu} is always satisfied.
\end{proposition}
\begin{proof}
    As Theorem \ref{thm: Cohen-Macaulay Matusevich Yu} is a perfect classification of Cohen-Macaulayness, and normality implies Cohen-Macaulayness by the theorem of Hochster (Theorem \ref{thm: Hochster}), condition (ii)' of Theorem \ref{thm: Cohen-Macaulay Matusevich Yu} must especially hold for normal semigroups.

    By Corollary \ref{cor: conv equal normal conv} and the following remark, the collection $\cA$ of facet defining hyperplanes of $\NN B$ equals that of the normalization $\NN \overline{B}$. The condition (ii)' only depends on $\cA$ and as it has to hold for $\NN\overline{B}$ it must also hold for $\NN B$.
\end{proof}
This means that for semigroups $Q$ coming from one-loop Feynman graphs, the ring $k[Q]$ is Cohen-Macaulay if and only if $\widetilde{Q}=Q$. The description of $\widetilde{Q}$ uses the facets of $Q$, as we will see in the next section these have a combinatorial description.
\subsection{Facet description}
In \cite{ohsugi1998normal} the authors not only provide necessary and sufficient conditions for normality but also give the facet description of the edge polytope. We begin with some definitions.

\begin{definition}[Regular vertex]
A vertex $i\in\{0,\ldots,d\}$ of a graph $H$ is said to be \textit{regular} in $H$ if every connected component of $H_{\{0,\cdots,n\}\setminus{i}}$ has at least one odd cycle.
\end{definition}
Let $N(H;i)$ denote the set of vertices in $H$ connected to $i$ by either an edge or a loop.
\begin{definition}[Independent subset]
 A non-empty subset $T\subset\{0,\dots,d\}$ is \textit{independent} in $H$ if $N(H;i)$ satisfies $N(H;i)\cap T=\emptyset$ for every $i\in T$.      
\end{definition}
This means that no edge $\{i,j\}$ with $i,j\in T$ belongs to $H$ no loop at $i\in T$ belongs to $H$.
\begin{definition}[Bipartite graph]
    A graph $H$ is called \textit{bipartite} if its vertex set $V(H)$ has a partition $V(H)=V_1\sqcup V_2$ with $V_1,V_2\neq \emptyset$ and $V_1\cap V_2=\emptyset$ such that for each edge $\{i,j\}$ in $H$ has $i\in V_1$ and $j\in V_2$. In other words a graph $H$ is \textit{bipartite} if and only if it has no odd cycle.
\end{definition}
Let $N(H;T)$ denote the union $\cup_{i\in T}N(H;i)$, then the bipartite graph \textit{induced} by $T$ in $H$ is defined as the bipartite graph with vertex set $T\sqcup N(H;T)$ and edges $\{i,j\}$ in $H$ such that $i\in T$ and $j\in N(H;T)$.
\begin{definition}[Fundamental set]
    A non-empty subset $T\subset\{0,\ldots,d\}$ is \textit{fundamental} in $H$ if
    \begin{itemize}
        \item[(a)] $T$ is independent in $H$ and the bipartite graph induced by $T$ in $H$ is connected, and
        \item[(b)] either $T\sqcup N(H;T)=\{0,\ldots, d\}$ or every connected component of the subgraph $H_{\{0,\ldots,d\}\setminus(T \sqcup N(H;T)}$ contains an odd cycle.
    \end{itemize}
\end{definition}
To describe the facets of the edge polytope we need to set of hyperplanes:
\begin{equation}
\begin{split}
    \mathscr{H}_i&:=\{(x_0,\dots,x_{d})\in \mathbb{R}^{d+1}\, |\, x_i=0\}, \\
    \mathscr{H}_T&:=\{(x_0,\dots,x_{d})\in \mathbb{R}^{d+1}\, |\, \sum_{i\in T}x_i=\sum_{j\in N(H;T)}x_j\}.
\end{split}
\end{equation}
With these hyperplanes we have the following facet description.
\begin{theorem}[{\cite[Theorem 1.7]{ohsugi1998normal}}]\label{thm: facets edge polytope}
    Let $H$ be a finite connected graph on the vertex set $\{0,\dots,d\}$ allowing loops and having no multiple edges with $H$ having at least one odd cycle. The set of facets of the edge polytope of $H$ are all the $\mathscr{H}_i$ for which $i$ is regular in $H$ and all the $\mathscr{H}_T$ for which $T$ is fundamental in $H$.
\end{theorem}
Using this theorem we get a very simple description of the facets of $\NN A$. When all internal masses are non-zero, every vertex of the kinematic graph is a regular vertex and the only fundamental set is $T=\{0\}$. For a one-loop Feynman graph with $d$ edges we get the $d+2$ facets:
\begin{equation}\label{eq: facets all massive}
\begin{split}
    \mathscr{H}_i&=\{(x_0,\dots,x_{d})\in \mathbb{R}^{d+1}\, |\, x_i=0\},\ i=0,1,\ldots, d,\\
    \mathscr{H}_{\{0\}}&=\{(x_0,\dots,x_{d})\in \mathbb{R}^{d+1}\, |\, x_0=\sum_{j=1}^dx_j\}.
\end{split}
\end{equation}
\section{Massive one-loop Feynman integrals}\label{sec: massive one loop}
In this section, we prove the main theorem of the paper, Theorem \ref{massivetheorem}. This theorem gives a necessary and sufficient condition for the semigroup ring of a one-loop integral with massive internal edges to be Cohen-Macaulay. Together with this theorem, we provide a simple sufficient condition for Cohen-Macaulayness, Theorem \ref{thm: massive sufficient not CM}, and also prove that a special subset of massive integrals are Cohen-Macaulay, Theorem \ref{thm: CM only loops}. The latter is of independent interest as the proof strategy is different from Theorem \ref{thm: massive sufficient not CM} and \ref{massivetheorem}.

Before proving the theorems, we start with the simplest example of a massive one-loop integral (note that the massive tadpole has an empty toric variety), the massive one-loop bubble. With massive edges there are only two different $A$-matrices for the bubble: 
\begin{equation}
    A_1=\begin{pmatrix}
        1&1&1&1&1\\
        1&0&1&2&0\\
        0&1&1&0&2
    \end{pmatrix},\quad A_2=\begin{pmatrix}
        1&1&1&1\\
        1&0&2&0\\
        0&1&0&2
    \end{pmatrix},
\end{equation}
where the first is for $m_1^2+m_2^2-p^2\neq 0$ and the latter for $m_1^2+m_2^2-p^2= 0$. The kinematic graphs for these cases are shown in Figure \ref{fig: generic massive bubble} and \ref{fig: special massive bubble} where it is easy to see that $A_1$ satisfies the odd cycle condition (Definition \ref{def:oddCycle}) so $\NN A_1$ is normal by Theorem \ref{thm: odd-cycle}. By the same theorem it follows that $\NN A_2$ is not normal as there is no edge connecting the two loops in Figure~\ref{fig: special massive bubble}. 
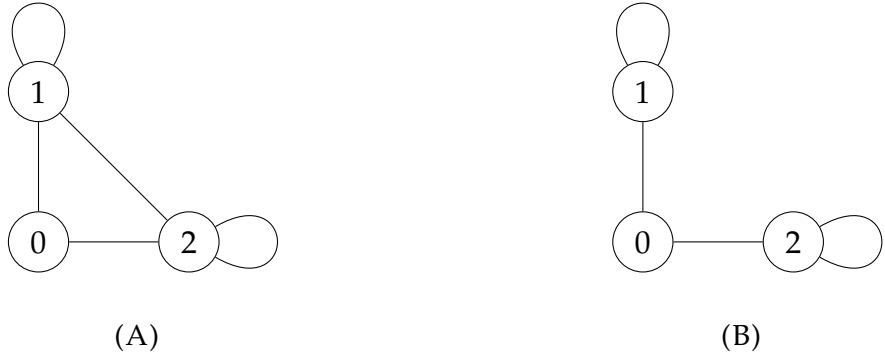
\begin{figure}[t]
    \begin{subfigure}[t]{0.45\textwidth}
    \centering
        \begin{tikzpicture}
    \node[circle, draw, minimum size=0.8cm] (0) at (0,0) {0};
    \node[circle, draw, minimum size=0.8cm] (1) at ({2*cos(90)}, {2*sin(90)}) {1};
    \node[circle, draw, minimum size=0.8cm] (2) at ({2*cos(0)}, {2*sin(0)}) {2};

    \draw (0) -- (1);
    \draw (0) -- (2);
    \draw (1) -- (2);

    \draw (1) to[out=60, in=120, looseness=8] (1);
    \draw (2) to[out=30, in=330, looseness=8] (2);

\end{tikzpicture}
\caption{}\label{fig: generic massive bubble}
    \end{subfigure}\hspace*{0.5cm}
    \begin{subfigure}[t]{0.45\textwidth}
    \centering
        \begin{tikzpicture}
    \node[circle, draw, minimum size=0.8cm] (0) at (0,0) {0};
    \node[circle, draw, minimum size=0.8cm] (1) at ({2*cos(90)}, {2*sin(90)}) {1};
    \node[circle, draw, minimum size=0.8cm] (2) at ({2*cos(0)}, {2*sin(0)}) {2};

    \draw (0) -- (1);
    \draw (0) -- (2); 

    \draw (1) to[out=60, in=120, looseness=8] (1);
    \draw (2) to[out=30, in=330, looseness=8] (2);
\end{tikzpicture}
   \caption{}\label{fig: special massive bubble}\end{subfigure}\caption{(A) Kinematic graph for massive bubble with $m_1^2+m_2^2-p^2\neq 0$ which satisfies the odd cycle condition. (B) Kinematic graph for massive bubble with $m_1^2+m_2^2-p^2=0$. This graph does not satisfy the odd cycle condition}
\end{figure}

However, $k[\NN A_1]$ is still Cohen-Macaulay even if it is not normal. We will now prove that this is not a special case for the bubble, but in fact, every massive one-loop integral with $m_i^2+m_j^2-p(F_{ij})^2=0$ for all $i$ and $j$ are Cohen-Macaulay.

Let $A$ be the $(d+1)\times 2d$ matrix 
\begin{equation}\label{eq: A matrix only loops}
    A=\left(
    \begin{array}{c|c}
        1\ \cdots\ 1  & 1\ \ \ \cdots\ \ 1 \\
        \mathbb{1}_{d\times d} & 2\cdot\mathbb{1}_{d\times d}
    \end{array}
    \right)
\end{equation}
which correspond to the kinematic graph is Figure \ref{fig:bipartite graph}. In order to prove that $k[\NN A]$ is Cohen-Macaulay we show that the Gröbner basis of the toric ideal $I_A$ with lexicographic ordering is easy to write down. The initial ideal in this ordering is homeomorphic to a monomial ideal that is the edge ideal of an bipartite graph satisfying a Cohen-Macaulay condition by Herzog and Hibi \cite[Corollary 9.1.14]{Herzog2011book}, originally proved in \cite[Theorem~3.4]{Herzog2005}. 

Let $R=k[x_1,\ldots,x_d,y_1,\ldots,y_d]$ be a polynomial ring with lexicographic ordering $x_1>\cdots>x_d>y_1>\cdots>y_d$ corresponding to the columns of $A$. The following set of binomials form a Gröbaner basis of $I_A$
\begin{equation}
    \{x_i^2y_j-x_j^2y_i\,|\,i,j\in\{1,\ldots,d\},\ i<j\}\subseteq I_A.
\end{equation}
The matrix $SA$ is essentially a \emph{Lawrence lifting} of the row-matrix $(1,\ldots,1)$ of $d$ ones. To be precise we have that $\mathbf{u}\in\mathrm{ker}((1,\ldots,1))\iff (2\mathbf{u}|-\mathbf{u})\in\mathrm{ker}(A)$ so we have
\begin{equation}
    I_A=\avg{x^{2\mathbf{u}_+}y^{\mathbf{u}_-}-x^{2\mathbf{u}_-}y^{\mathbf{u}_+}\,|\,\mathbf{u}\in\mathrm{ker}((1,\ldots,1))}.
\end{equation}
It now follows from \cite[Theorem 7.1]{Sturmfels1996} that the full set of $\{x_i^2y_j-x_j^2y_i\}$ is not just a Gröbner basis but also a universal and Graver basis. For this special $A$, it is also the set of circuits.

The initial ideal with lexicographic ordering is thus
\begin{equation}\label{eq: initial ideal A only loops}
    \mathrm{in}_\prec(I_A)=\avg{ x_i^2y_j\ |\ i,j\in\{1,\ldots,d\},\ i<j}.
\end{equation}
The initial ideal is important as a homogeneous ideal $I_A$ is Cohen-Macaulay if any of its initial ideals $\mathrm{in}_\prec(I_A)$ are Cohen-Macaulay. Let $J=\sqrt{\mathrm{in}_\prec(I_A)}$ denote the radical of this monomial ideal. If $J$ is Cohen-Macaulay, then so is $\mathrm{in}_\prec(I_A)$ since the map $\varphi:R\to R$ given by the substitution $\varphi(x_i)=x_i^2$ for $i=1,\ldots,d-1$ gives $\varphi(J)R=\mathrm{in}_\prec(I_A)$ and is a flat $k$-algebra homomorphism. All we now have to show is that $J$ is Cohen-Macaulay. To do this we use the following classification by Herzog and Hibi:

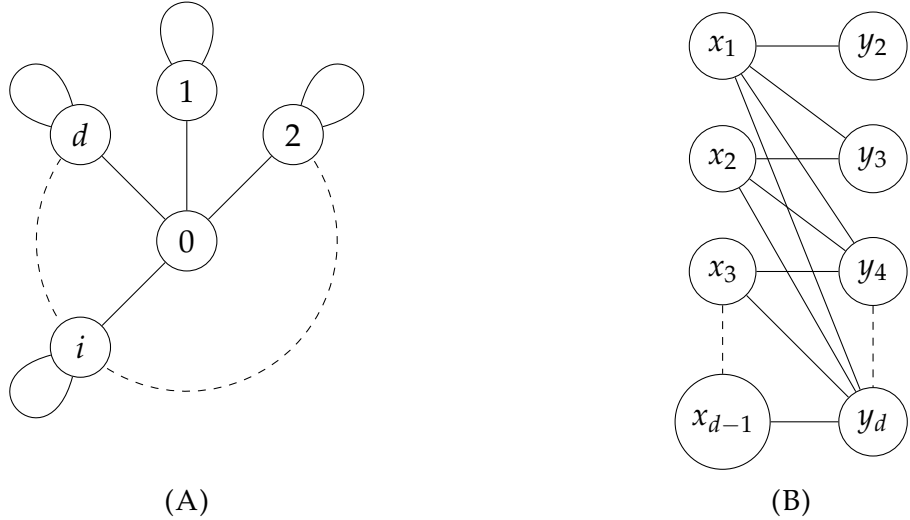
\begin{figure}
    \centering
        \begin{subfigure}[t]{0.45\textwidth}
    \centering
        \begin{tikzpicture}
    \node[circle, draw, minimum size=0.8cm] (0) at (0,0) {0};
    \node[circle, draw, minimum size=0.8cm] (1) at ({2*cos(90)}, {2*sin(90)}) {1};
    \node[circle, draw, minimum size=0.8cm] (2) at ({2*cos(45)}, {2*sin(45)}) {2};
    \node[circle, draw, minimum size=0.8cm] (i) at ({2*cos(225)}, {2*sin(225)}) {\(i\)};
    \node[circle, draw, minimum size=0.8cm] (d) at ({2*cos(135)}, {2*sin(135)}) {\(d\)};

    \draw (0) -- (1);
    \draw (0) -- (2);
    \draw (0) -- (i);
    \draw (0) -- (d);
    \draw[dashed] (238:2) arc (238:393:20mm);
    \draw[dashed] (210:2) arc (210:147:20mm);
    
    \draw (1) to[out=60, in=120, looseness=8] (1);
    \draw (2) to[out=15, in=75, looseness=8] (2);
    \draw (i) to[out=195, in=255, looseness=8] (i);
    \draw (d) to[out=105, in=165, looseness=8] (d);

\end{tikzpicture}
\caption{}\label{fig: only loops}
    \end{subfigure}\hspace*{0.5cm}
        \begin{subfigure}[t]{0.45\textwidth}
    \centering
        \begin{tikzpicture}
    \node[circle, draw, minimum size=0.8cm] (0) at (-1,2.5) {\(x_1\)};
    \node[circle, draw, minimum size=0.8cm] (1) at (1,2.5) {\(y_2\)};
    \node[circle, draw, minimum size=0.8cm] (2) at (-1, 1) {\(x_2\)};
    \node[circle, draw, minimum size=0.8cm] (3) at (1, 1) {\(y_3\)};
    \node[circle, draw, minimum size=0.8cm] (4) at (-1, -0.5) {\(x_3\)};
    \node[circle, draw, minimum size=0.8cm] (5) at (1, -0.5) {\(y_4\)};
        \node[circle,draw] (6) at (-1, -2.5) {\(x_{d-1}\)};
    \node[circle, draw, minimum size=0.8cm] (7) at (1, -2.5) {\(y_d\)};

\draw (0) -- (1);
\draw (0) -- (3);
\draw (0) -- (5);
\draw (0) -- (7);
\draw (2) -- (3);
\draw (2) -- (5);
\draw (2) -- (7);
   \draw (4) -- (5);
\draw (4) -- (7);
\draw (6) -- (7);
\draw[dashed] (4)--(6);
\draw[dashed] (5)--(7);

\end{tikzpicture}
\caption{} \label{fig:bipartite graph}
\end{subfigure}
    \caption{(A) The kinematic graph for \eqref{eq: A matrix only loops}, this correspond to $m_i^2+m_j^2-p(F_{ij})^2=0$ for all $i$ and~$j$. (B) Bipartite graph for the monomial ideal $\sqrt{\mathrm{in}_\prec(I_A)}$.}
\end{figure}
\begin{theorem}[{\cite[Corollary 9.1.14]{Herzog2011book}}]\label{thm: CM Herzog Hibi}
    Let $G$ be a bipartite graph on the vertex set $V\bigsqcup V',$ where $V=\{x_1,\ldots,x_d\}$ and $V'=\{y_1,\ldots,y_d\}$, and suppose that $G$ satisfies the conditions:
    \begin{itemize}
        \item[(i)] $\{x_i,y_i\}$ is an edge of $G$ for all $1\le i\le d$,
        \item[(ii)] if $\{x_i,y_j\}$ is an edge of $G$, then $i\le j$.
    \end{itemize}
    Then $G$ is Cohen-Macaulay if and only if the following condition $(iii)$ is satisfied:
    \begin{itemize}
        \item[(iii)] If $\{x_i,y_j\}$ and $\{x_j,y_k\}$ are edges of $G$ with $i<j<k$, then $\{x_i,y_k\}$ is an edge of $G$.
    \end{itemize}
\end{theorem}
Using this theorem it follows directly that massive one-loop integrals with $m_i^2+m_j^2-p(F_{ij})^2=0$ for all $i$ and $j$ are Cohen-Macaulay:
\begin{theorem}\label{thm: CM only loops}
    The semigroup ring $k[\NN A]$ with $A$ as in \eqref{eq: A matrix only loops} is Cohen-Macaulay.
\end{theorem}
\begin{proof}
    We have constructed an initial ideal for the toric ideal in \eqref{eq: initial ideal A only loops}. Its radical correspond to the edge ideal of the graph depicted in Figure \ref{fig:bipartite graph}. Relabelling $y_i\to y_{i-1}$ in this graph we see that this graph satisfies condition (i), (ii), and (iii) in Theorem \ref{thm: CM Herzog Hibi} and is therefore Cohen-Macaulay.
\end{proof}

For Feynman integrals with three or more internal edges, we have cases when the semigroup ring is not Cohen-Macaulay. The simplest massive example is the triangle with $A$-matrix
\begin{equation}\label{eq: AmatrixTrianglenotCM}
    A=\begin{pmatrix}
        1&1&1&1&1&1&1&1\\
        1&0&0&1&0&2&0&0\\
        0&1&0&0&1&0&2&0\\
        0&0&1&1&1&0&0&2
    \end{pmatrix}.
\end{equation}
The kinematic graph for this is shown in Figure \ref{fig: nonCM three mass} and the Newton polytope of the Lee-Pomeransky polynomial in Figure \ref{fig: nonCMpolytope3mass}. This kinematic graph does not satisfy the odd cycle condition so the affine semigroup $\NN A$ is certainly not normal (see also \cite[Theorem 5.4]{Dlapa:2023cvx}). The length of the minimal free resolution is 5, which is larger than the Auslander-Buchsbaum bound of 4, so $k[\NN A]$ is not Cohen-Macaulay. By Proposition \ref{prop: (ii) always satisfied} condition (ii)' is always satisfied in Theorem \ref{thm: Cohen-Macaulay Matusevich Yu} so Cohen-Macaulayness can also be tested by checking if $\bigcap_{i=1}^5(\NN A-\NN F_i)=\NN A$ where $F_i$ are the facets of $A$. 

\begin{figure}
    \centering
    \begin{subfigure}[t]{0.5\textwidth}
    \centering
        \begin{tikzpicture}[scale=1]
    \node[circle, draw, minimum size=0.8cm] (0) at (0,0) {0};
    \node[circle, draw, minimum size=0.8cm] (1) at ({2*cos(90)}, {2*sin(90)}) {1};
    \node[circle, draw, minimum size=0.8cm] (2) at ({2*cos(-30)}, {2*sin(-30)}) {2};
    \node[circle, draw, minimum size=0.8cm] (3) at ({2*cos(-150)}, {2*sin(-150)}) {3};

    \draw (0) -- (1);
    \draw (0) -- (2);
    \draw (0) -- (3);
    
    \draw (2) -- (3);
    \draw (1) -- (3);

    \draw (1) to[out=60, in=120, looseness=8] (1);
    \draw (2) to[out=300, in=360, looseness=8] (2);
    \draw (3) to[out=180, in=240, looseness=8] (3);

\end{tikzpicture}
\vspace*{-0.68cm}
\caption{}\label{fig: nonCM three mass}
    \end{subfigure}\hspace*{0.5cm}%
        \begin{subfigure}{0.5\textwidth}
        \begin{tikzpicture}%
	[x={(-1cm, -0.395971cm)}, 
	y={(1cm, -0.401902cm)}, 
	z={(0.017418cm, 1cm)},
	scale=2.000000,
	back/.style={loosely dotted, thin},
	edge/.style={color=black, thick},
	facet/.style={fill=green,fill opacity=0.200000},
	vertex/.style={inner sep=1pt,circle,draw=blue!25!black,fill=blue!75!black,thick},
    baseline=-\the\dimexpr\fontdimen22\textfont2\relax]
%
%
\coordinate (0.00000, 0.00000, 1.00000) at (0.00000, 0.00000, 1.00000);
\coordinate (0.00000, 0.00000, 2.00000) at (0.00000, 0.00000, 2.00000);
\coordinate (0.00000, 1.00000, 0.00000) at (0.00000, 1.00000, 0.00000);
\coordinate (2.00000, 0.00000, 0.00000) at (2.00000, 0.00000, 0.00000);
\coordinate (0.00000, 2.00000, 0.00000) at (0.00000, 2.00000, 0.00000);
\coordinate (1.00000, 0.00000, 0.00000) at (1.00000, 0.00000, 0.00000);
\draw[edge,back] (0.00000, 0.00000, 1.00000) -- (0.00000, 0.00000, 2.00000);
\draw[edge,back] (0.00000, 0.00000, 1.00000) -- (0.00000, 1.00000, 0.00000);
\draw[edge,back] (0.00000, 0.00000, 1.00000) -- (1.00000, 0.00000, 0.00000);
\draw[edge,back] (0.00000, 1.00000, 0.00000) -- (0.00000, 2.00000, 0.00000);
\draw[edge,back] (0.00000, 1.00000, 0.00000) -- (1.00000, 0.00000, 0.00000);
\draw[edge,back] (2.00000, 0.00000, 0.00000) -- (1.00000, 0.00000, 0.00000);
\node[vertex] at (0.00000, 1.00000, 0.00000)     {};
\node[vertex] at (1.00000, 0.00000, 0.00000)     {};
\node[vertex] at (0.00000, 0.00000, 1.00000)     {};
\node[vertex] at (0.00000, 1.00000, 1.00000)     {};
\node[draw, circle,inner sep=2pt] at (1.00000, 1.00000, 0.00000)     {};
\node[vertex] at (1.00000, 0.00000, 1.00000)     {};
\fill[facet] (0.00000, 2.00000, 0.00000) -- (0.00000, 0.00000, 2.00000) -- (2.00000, 0.00000, 0.00000) -- cycle {};
\draw[edge] (0.00000, 0.00000, 2.00000) -- (2.00000, 0.00000, 0.00000);
\draw[edge] (0.00000, 0.00000, 2.00000) -- (0.00000, 2.00000, 0.00000);
\draw[edge] (2.00000, 0.00000, 0.00000) -- (0.00000, 2.00000, 0.00000);
\node[vertex] at (0.00000, 0.00000, 2.00000)     {};
\node[vertex] at (2.00000, 0.00000, 0.00000)     {};
\node[vertex] at (0.00000, 2.00000, 0.00000)     {};
\end{tikzpicture}
\caption{}\label{fig: nonCMpolytope3mass}
    \end{subfigure}
\caption{(A) Kinematic graph of a three point Feynman integral with a non-CM semi group $k[H]$ and (B) the corresponding polytope. Note that the missing edge between the vertices 1 and 2 corresponds to the missing point i.e the hole in the polytope in (B).}
\end{figure}
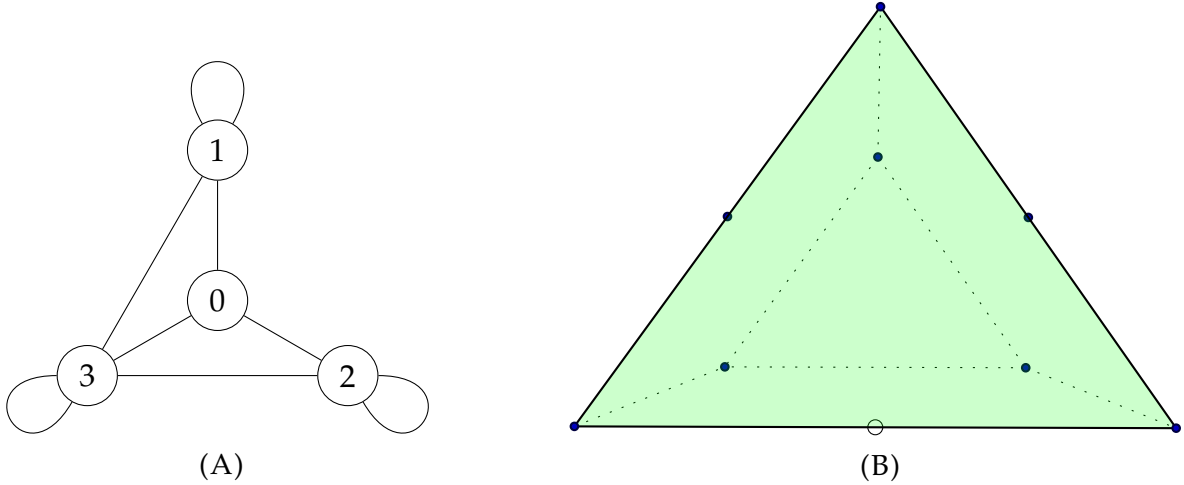
The only way $\widetilde{Q}=Q$ can fail is if a hole $\mathbf{b}_j$ of $Q$ lie in every $\NN A-\NN F_i$, and this is an \emph{integer linear program}: 
\begin{equation}\label{eq: ILP}
    \left(A|-F_i\right)\boldsymbol{x}=\mathbf{b}_j,\quad \xx\ge 0,
\end{equation}
for every facet $i$. We now get the following corollary of Proposition \ref{prop: (ii) always satisfied}:
\begin{corollary}\label{cor: ilp not CM}
    $\widetilde{Q}\neq Q$, i.e. $k[Q]$ is \textbf{not} Cohen-Macaulay, if and only if their exists a hole $\mathbf{b}_j$ of $Q$ such that integer linear program \eqref{eq: ILP} has a solution for every facet $F_i$ of $Q$.
\end{corollary}
For the $A$-matrix in \eqref{eq: AmatrixTrianglenotCM}, the semigroup $\NN A$ has a hole $\mathbf{b}=(1,1,1,0)^T$, corresponding to the missing edge between vertex 1 and 2 in Figure \ref{fig: nonCM three mass}. It is easy to check explicitly that all five integer linear programs \eqref{eq: ILP} have a solution, meaning that $\widetilde{Q}\neq Q$ and $k[\NN A]$ is not Cohen-Macaulay, which we of course already knew. 

The graph in Figure \ref{fig: nonCM three mass} is very small and it turns out that any kinematic graph with this graph as a subgraph is not Cohen-Macaulay: 
\begin{theorem}\label{thm: massive sufficient not CM}
    Let $H$ be a kinematic graph with all internal masses being non-zero and $B$ its edge matrix. If, up to relabelling, $B$ contains the graph in Figure \ref{fig: nonCM three mass} as a subgraph, then the semigroup ring $k[\NN B]$ is not Cohen-Macaulay.
\end{theorem}
\begin{proof}
    Let $B'$ denote the edge matrix for the graph \ref{fig: nonCM three mass}, i.e.
    \begin{equation}
        B'=\begin{pmatrix}
            1&1&1&0&0&0&0&0\\
            1&0&0&1&0&2&0&0\\
            0&1&0&0&1&0&2&0\\
            0&0&1&1&1&0&0&2
        \end{pmatrix}.
    \end{equation}
Assume $H$ contains \ref{fig: nonCM three mass} as a subgraph, then without loss of generality we can choose a labelling of the vertices coinciding with \ref{fig: nonCM three mass} so the edge matrix of this graph has the form
    \begin{equation}\label{eq: B massive proof}
        B=\left(
        \begin{array}{c|c}
            B' & * \\
            \hline
             0& *
        \end{array}
        \right).
    \end{equation}
    The semigroup $\NN B'$ has a hole at $\mathbf{b}'=(0,\,1,\,1,\,0)^T$, meaning that $\NN B$ has a hole at $(0,\,1,\,1,\,0,\ldots,0)^T$. We have already seen that the integer linear program \eqref{eq: ILP} for $B'$ has a solution on all five  and we will see that these solutions lift to solutions for $B$.

    Since all internal masses are non-zero, the facets of $\NN B$ are simply given by \eqref{eq: facets all massive}, for the facets of $B$ we use the notation
    \begin{equation*}
        F_i=\scrH_i\cap B,\qquad F_{\{0\}}=\scrH_{\{0\}}\cap B.
    \end{equation*}
    By having the matrix $B$ on the special form \eqref{eq: B massive proof}, we have the following inclusions thought of as sets of lattice points
    \begin{equation*}
    \left(\begin{array}{c}
         F'_i \\
         \hline
         0 
    \end{array}
    \right)\subset F_i,\ i=0,\ldots,3\ \mathrm{or}\ \{0\},\qquad
        \left(\begin{array}{c}
            B' \\
            \hline
             0
        \end{array}\right)\subset F_i,\ i=4,\ldots,d.
    \end{equation*}
    We will construct solutions for these two cases separately.
    
    For the first case, we can assume without loss of generality that the augmented matrix $B_i=(B|-F_i)$ has the form
    \begin{equation*}
        B_i=\left(\begin{array}{c|c|c|c}
        B'&*&-F_i'&*\\
        \hline
        0&*&0&*
        \end{array}
        \right),\ i=0,\ldots,3\ \mathrm{or}\ \{0\}.
    \end{equation*}
    If $\xx'$ is a solution to the linear program $B_i'\xx'=(B'|-F_i')\xx'=\mathbf{b}',\ \xx'\ge 0$, then it lifts to a solution $\xx$ of the program $B_i\xx=\mathbf{b},\ \xx\ge 0,$ by splicing in zeros in the positions corresponding to the $*$s in $B_i$.
    
    For the second case, we can assume without loss of generality that the augmented matrix $B_i$ has the form
    \begin{equation*}
        B_i=\left(\begin{array}{c|c|c|c}
        B'&*&-B'&*\\
        \hline
        0&*&0&*
        \end{array}
        \right),\ i=4,\ldots,d.
    \end{equation*}
    This means again that any solution $\xx'$ can be lifted to a solution $\xx$ by splicing in zeros in the correct positions.

    We conclude that the integer linear programs $B_i\xx=\mathbf{b},\ \xx\ge 0,$ with $B_i=(B|-F_i)$ has a solution for every facet $i=0,\ldots,n$ or $\{0\}$. By Corollary \ref{cor: ilp not CM} it follows that $k[\NN B]$ is not Cohen-Macaulay.
\end{proof}
\begin{definition}
    The \textit{generalised linear program} of a kinematic graph is the linear program in which the vectors of every edge can be added and subtracted i.e. $(B|-B)\boldsymbol{x}=\mathbf{b}$.
\end{definition}
The following Lemma \ref{equivalence} allows for identifying solutions to the linear program with paths in the kinematic graph, yielding a graphical way of determining if the linear programs are feasible and thereby if the associated semigroup ring $k[\NN B]$ is Cohen-Macaulay. 
\begin{lemma}\label{equivalence}
    There is a bijective correspondence between solutions of the generalised linear program and paths\footnote{We do not require all vertices to be distinct in a path.} of odd length between massive vertices in the kinematic graph (that will be called \textit{connecting paths}).
\end{lemma}
\begin{proof}
    We construct a canonical bijection that maps any connecting path to exactly one solution of the corresponding linear program and whose inverse maps any solution of the linear program to one connecting path. 
    
    Let $H$ be a kinematic graph on the vertex set $\{0,\ldots,d\}$  with edge matrix $B=(\boldsymbol{{\mathfrak{b}}}_1,\dots,\boldsymbol{{\mathfrak{b}}}_n)$, and with a hole $\mathbf{b}$: 
    \begin{equation}
    \mathbf{b}=(0, \dots,0,1,0, \dots,0,1,0,\dots,0)^T\in\mathbb{N}^{d+1}
    \end{equation}
    such that $\mathbf{b}$ corresponds to the missing edge $\{u,v\}$, between the two massive vertices $u$ and $v$, i.e. the ones in $\mathbf{b}$ are at position $u$ and $v$. 
    Let $P=\{e_1,\dots,e_m\}$ be a connecting path, connecting $u$ with $v$, where $e_i$ labels the $i$th edge in the path. There is a canonical bijective correspondence between an edge $e_i\in P$ and an index $k\in\{1,\ldots,n\}$ labelling a column of $B$:   
    \begin{equation}
    e_i\leftrightarrow\boldsymbol{\mathfrak{b}_k}\leftrightarrow k
    \end{equation}
    by the identification of every edge in the kinematic graph $H$ with a column in $B$.
    Using this we construct a map that identifies every edge in the path with a unit basis vector $\mathbf{e}_j\in \RR^{2n}$ as:
    \begin{equation}\label{eq: ei map}
        e_i\mapsto\begin{cases}
            \mathbf{e}_k, \ &\text{for $i$ odd,}\\
            \mathbf{e}_{n+k}, \ &\text{for $i$ even,}
        \end{cases}
    \end{equation}
    such that the whole path $P$ gets mapped to:
    \begin{equation}\label{eq: map}
        P\mapsto\sum_{i\ \mathrm{odd}}\mathbf{e}_i+\sum_{i\ \mathrm{even}}\mathbf{e}_{i+n}=:\boldsymbol{x}
    \end{equation}
    where both sums are from $1$ to $m$. With this $\boldsymbol{x}$ we get for the left hand side of the linear program
    \begin{equation}\label{eq: cancellations in path}
    (B|-B)\boldsymbol{x}=\sum_{i=1}^n\boldsymbol{\mathfrak{b}}_ix_i-\sum^{2n}_{i=n+1}\boldsymbol{{\mathfrak{b}}}_{i-n}x_i.  
    \end{equation}
    By assumption, the vertices $u$ and $v$ are different so in the path $P$ every vertex has even degree except $u$ and $v$. Equation \eqref{eq: cancellations in path} then gives that the contributions from all internal vertices cancel and we are left with
    \begin{equation}
        (B|-B)\boldsymbol{x}=\mathbf{b}.
    \end{equation}
    Hence, every connecting path maps to a solution of the linear program.
    
    For the inverse, let us assume that there is a solution to the linear program $\boldsymbol{x}\in\NN^{2n}$. Each $\boldsymbol{x}$ can be written in terms of the basis vectors using \eqref{eq: map} and then mapped to a family of potentially repeated edges $\{e_1,\ldots,e_m\}$. All that remains is to show that this is a connecting path. 
    
    Since $\boldsymbol{x}$ is a solution to the linear program, there exists at least one edge in this set containing the vertex $u$ and at least one edge containing $v$. Since $\mathbf{b}$ has non-zero elements at entries $u,v$, the degree of these vertices must be odd. Similarly, since all other entries of $\mathbf{b}$ are zero, there can be no other vertex with odd degree. We conclude that $P$ is a connecting path of odd length, as entries in \eqref{eq: cancellations in path} cancel pairwise.
\end{proof}
\begin{remark}
    Lemma \ref{equivalence} can be used to determine the Cohen-Macaulayness of a given one-loop Feynman integral by identifying the connecting paths of the holes in the kinematic graph and verifying if there is at least one connecting path for which all alternating edges (starting from the second edge in the path) are contained in the facet of the corresponding linear program. In that case, the linear program of that facet is feasible, and if the same holds for all linear programs of the graph, the associated semigroup ring is not Cohen-Macaulay. Conversely, if there is a facet that does not contain alternating edges for any connecting path, the associated semigroup ring is Cohen-Macaulay. Consider for example the graph in Figure \ref{fig: special massive bubble}: All of the connecting paths of the hole $\mathbf{b}=(0,1,1)^T$, contain the edges $1-0$ and $0-2$, following each other. Since the $\mathscr{H}_0$ facet contains no edges connected to the zero-vertex, it does not contain $1-0$ and $0-2$, which means the corresponding linear program is infeasible and the associated semi group ring is Cohen-Macaulay.
\end{remark}
\begin{theorem}\label{massivetheorem}
    The affine semi group ring of a fully massive one-loop Feynman integral is Cohen-Macaulay if and only if all the holes are between different connected components of the reduced kinematic graph $H_{[0,d]\setminus0}$.
\end{theorem}
If there are no holes, the semigroup is normal, so for the proof we assume that there exists at least one hole.
\begin{proof}
    The facets of a fully massive one-loop integral are given by \eqref{eq: facets all massive}. For all facets $\scrH_{1,\ldots,d\{0\}}$ there is a path (of length three) that solves the linear program. 
    
    Every facet $\mathscr{H}_{i}$ with $i\in\{1,\dots,d\}$ corresponds to the reduced graph kinematic graph $H_{[0,d]\setminus i}$ from which follows that for every hole that exists in $H_{[0,d]\setminus i}$ there is a path through the $0$-vertex that connects the vertices of each side of the hole to each other. Every edge in the path connecting the two sides of the hole corresponds to vectors in $\scrH_i$. By Lemma~\ref{equivalence} this path generate a solution to the linear program. 
    
    The facet $\mathscr{H}_{\{0\}}$ only contains vectors from the $\mathcal{U}$-polynomial i.e. the $0$-vertex and the connections to the individual vertices, which allows for paths connecting all vertices, as the zero-vertex is connected to every vertex. 
    
    The only facet that can correspond to a graph with more than one connected component is $\mathscr{H}_0$ since the zero vertex is set to zero here. If the hole lies in a connected component of the reduced kinematic graph $H_{[0,d]\setminus0}$, there still exists a path in the graph connecting the two sides of the hole and thereby solving the corresponding linear program. If the hole is between two disjoint connected components of $H_{[0,d]\setminus0}$, the only way the linear program could have a solution would be if the two components individually summed to a vector with zeros everywhere except for a one at the vertex for the hole. This is impossible since every column of the matrix in the linear program sums to $\pm2$, with which one cannot generate any odd number.
\end{proof}
\section{On general one-loop Feynman integrals}\label{sec: general one-loop}
In this section we include massless particles and make use of the results from the previous section to prove a criterium for Cohen-Macaulayness that is applicable to the most general one-loop Feynman integral and condenses to a full classification for up to three-point integrals.

In the following, we only consider Feynman integrals whose kinematic graph $H_{[0,d]}$ contains a massive subgraph $\widetilde{H}_{[0,j]}$ with more than one massive vertex. Here we assume a labelling of the vertices such that $1,\ldots,j$ are the massive vertices. Since Feynman integrals whose kinematic graphs do not contain massive subgraphs (or subgraphs with only one massive vertex) are normal \cite[Corollary 5.6]{Dlapa:2023cvx} and therefore Cohen-Macaulay. We consider the zero-vertex as part of the massive subgraphs, which means that these subgraphs make up one connected component, since the zero-vertex is adjacent to every vertex in the graph.

\begin{theorem}\label{generaltheorem}
    A kinematic graph $H_{[0,d]}$ containing a fully massive subgraph $\widetilde{H}_{[0,j]}$ is Cohen-Macaulay if the holes of $\widetilde{H}_{[0,j]}$ lie in different connected components of $H_{[0,d] \setminus 0}$.
\end{theorem}
\begin{proof}
    If the massless vertices are not part of a subgraph that connects the hole(s) in the reduced subgraph $H_{[0,d] \setminus 0}$ there is still no path that will solve the linear program $(A|-F_0)\boldsymbol{x}=\mathbf{b}$, where $F_0=A\cap\scrH_0$, by the same argument as presented in the proof of Theorem \ref{massivetheorem}. Thus, Cohen-Macaulayness is preserved by adding massless vertices that do not connect the hole(s) of a fully massive Cohen-Macaulay kinematic graph. If the massive subgraph is not Cohen-Macaulay, then we know by Theorem \ref{massivetheorem} that the hole(s) lie in the same connected component in $\widetilde{H}_{[0,j] \setminus 0}$, from which follows that they also lie in the same connected component of the bigger graph $H_{[0,d] \setminus 0}$.
\end{proof}
\begin{remark}\label{converseneednotbetrue}
    Note that the converse need not be true, as a non-Cohen-Macaulay kinematic graph that contains a fully massive Cohen-Macaulay subgraph and one extra massless vertex (that is not the zero vertex), can be non-Cohen-Macaulay for two reasons: either the linear program $(A|-F_0)\boldsymbol{x}=\mathbf{b}$ is now feasible (in which case the massless vertex connects the hole and we now have a full equivalence for the above theorem), or the linear program $(A|-F_T)\boldsymbol{x}=\mathbf{b}$ with $T$ being the added massless vertex, can now be feasible (or both). Since the latter case can occur independently of the first one, non-Cohen-Macaulay graphs whose massless vertex does not connect the hole(s) exist. See Figure \ref{fig: counterex} for an example of such a kinematic graph.
\end{remark}
\begin{center}
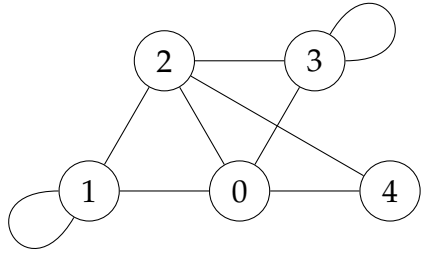
\begin{figure}[t]
        \begin{tikzpicture}
    \node[circle, draw, minimum size=0.8cm] (0) at (0,0) {0};
    \node[circle, draw, minimum size=0.8cm] (3) at ({2*cos(-300)}, {2*sin(-300)}) {3};
    \node[circle, draw, minimum size=0.8cm] (1) at ({2*cos(-180)}, {2*sin(-180)}) {1};
    \node[circle, draw, minimum size=0.8cm] (2) at ({2*cos(-240)}, {2*sin(-240)}) {2};
    \node[circle, draw, minimum size=0.8cm] (4) at ({2*cos(-360)}, {2*sin(-360)}) {4};
 
    \draw (0) -- (1);
    \draw (0) -- (2);
    \draw (0) -- (3);
    \draw (0) -- (4);

    \draw (2) -- (3);
    \draw (1) -- (2);
    \draw (4) -- (2);

    \draw (1) to[out=240, in=180, looseness=8] (1);
    \draw (3) to[out=0, in=60, looseness=8] (3);

\end{tikzpicture}
\caption{This example demonstrates the case in which the criterium of Theorem \ref{generaltheorem} does not provide an equivalence, which proves the statement of Remark \ref{converseneednotbetrue}. Note that in this kinematic graph there is a massive subgraph containing the vertex set $\{0,1,3\}$, and two additional massless vertices $2,4$. The vertex $2$ connects the two components $\{1\}$ and $\{3\}$ in $H_{[0,4]\setminus 0}$, but since the edges $1-2$, $2-3$,$1-0$ and $0-3$ are not in the facet $\mathscr{H}_{T=\{4\}}$, the corresponding linear program is infeasible and the graph is therefore Cohen-Macaulay.}\label{fig: counterex}
\end{figure}
\end{center}
\begin{corollary}\label{Corrollary on 3-pt-integrals}
    A kinematic graph $H_{[0,d]}$ with $d\leq 3$ is Cohen-Macaulay if and only if the holes of its fully massive subgraph $\widetilde{H}_{[0,j]}$ lie in different connected components of $H_{[0,d] \setminus 0}$.
\end{corollary}
\begin{center}
\begin{figure}[t]
        \begin{tikzpicture}
    \node[circle, draw, minimum size=0.8cm] (0) at (0,0) {0};
    \node[circle, draw, minimum size=0.8cm] (3) at ({2*cos(-300)}, {2*sin(-300)}) {3};
    \node[circle, draw, minimum size=0.8cm] (1) at ({2*cos(-180)}, {2*sin(-180)}) {1};
    \node[circle, draw, minimum size=0.8cm] (2) at ({2*cos(-240)}, {2*sin(-240)}) {2};

    \draw (0) -- (1);
    \draw (0) -- (2);
    \draw (0) -- (3);

    \draw (2) -- (3);
    \draw (1) -- (2);

    \draw (1) to[out=240, in=180, looseness=8] (1);
    \draw (3) to[out=0, in=60, looseness=8] (3);

\end{tikzpicture}
\caption{This graph displays the unique kinematic graph of a non-normal three point integral, that does not satisfy the criterium of Theorem \ref{generaltheorem}. As discussed in the proof of Corollary \ref{Corrollary on 3-pt-integrals}, this graphs' semigroup is not Cohen-Macaulay, as all its linear programs are feasible.}\label{fig: prooffig}
\end{figure}
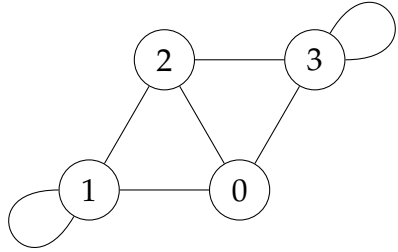
\end{center}
\begin{proof}
    This can be shown by a simple case distinction. All Feynman integrals with $0$ or $1$ massive vertices are normal and therefore Cohen-Macaulay. For no massless vertices this has already been proven in Theorem \ref{massivetheorem}. Thus the only remaining case to consider is $2$ massive vertices and $1$ massless vertex, see Figure \ref{fig: prooffig}. This case is straightforward, as the massless vertex is directly adjacent to the two massive vertices. Therefore, all edges in the connecting path through the massless vertex are in the facet $\mathscr{H}_T$ ($T$ being the massless vertex) making the linear program feasible and the whole graph no longer Cohen-Macaulay. Together with Theorem \ref{generaltheorem}, this yields the full equivalence.
\end{proof}
This theorem provides a tool to detect Cohen-Macaulayness for the most general one-loop integrals; it would of course be desirable to find a full classification.
\section*{Acknowledgments}
FT would like to thank Martin Helmer for many helpful discussions at the early stages of this work. FT was funded by Royal Society grant number URF\textbackslash R1\textbackslash 201473 for the majority of this work and gratefully acknowledges this support. The research of FT is also funded by the European Union ERC Synergy Grant MaScAmp 101167287. Views and opinions expressed are however those of the author(s) only and do not necessarily reflect those of the European Union or the European Research Council. Neither
the European Union nor the granting authority can be held responsible for them. 

KM is grateful for the support of the Studienstiftung des deutschen Volkes (German Academic Scholarship Foundation) and gratefully acknowledges the support of the MLP-scholarship.
\bibliographystyle{JHEP}
\bibliography{refs}

\end{document}